\documentclass[acmsmall, nonacm, natbib=false]{acmart}

\usepackage{microtype}

\usepackage{nicematrix,tikz}
\usepackage{adjustbox}

\usepackage{tikz}
\usepackage{pgfplots}
\usetikzlibrary {automata,positioning}
\usetikzlibrary {shapes.geometric}
\usetikzlibrary{arrows.meta}

\usepackage[
style=lncs,
backend=biber,
alldates=year,
sortcites=false
]{biblatex}

\usetikzlibrary{arrows.meta}
\usetikzlibrary{backgrounds}
\usepgfplotslibrary{patchplots}
\usepgfplotslibrary{fillbetween}
\pgfplotsset{%
    layers/standard/.define layer set={%
        background,axis background,axis grid,axis ticks,axis lines,axis tick labels,pre main,main,axis descriptions,axis foreground%
    }{
        grid style={/pgfplots/on layer=axis grid},%
        tick style={/pgfplots/on layer=axis ticks},%
        axis line style={/pgfplots/on layer=axis lines},%
        label style={/pgfplots/on layer=axis descriptions},%
        legend style={/pgfplots/on layer=axis descriptions},%
        title style={/pgfplots/on layer=axis descriptions},%
        colorbar style={/pgfplots/on layer=axis descriptions},%
        ticklabel style={/pgfplots/on layer=axis tick labels},%
        axis background@ style={/pgfplots/on layer=axis background},%
        3d box foreground style={/pgfplots/on layer=axis foreground},%
    },
}

\usepackage{algorithm}
\usepackage{algpseudocode}

\usepackage{xcolor}
\usepackage{tabularray}
\usepackage{array}

\usepackage{longtable}
\usepackage{booktabs}

\newcommand{\splt}{\mbox{\kern .1em $\diamond$\kern .1em}}

\newcommand{\texticon}[1]{\text{\scalebox{.7}{\kern 0.1em #1\kern .1em}}}

\newcommand{\automaton}[1]{{\ensuremath{\vphantom{t}\smash{\textit{#1}}}}}

\newcommand{\Tscissors}{\automaton{Scissors}}
\newcommand{\Tcounter}[1]{\ensuremath{\automaton{Cnt}_{#1}}}
\newcommand{\Taudio}{\automaton{Audio}}
\newcommand{\Taa}{\ensuremath{\automaton{Audio}_+}}
\newcommand{\Trecog}{\automaton{Sink}}
\newcommand{\Tgen}{\automaton{Src}}
\newcommand{\Rsplit}{\automaton{Splitting}}
\newcommand{\Rirred}{\automaton{Atom}}
\newcommand{\Tirred}{\automaton{AtomicF}}
\newcommand{\Tsinglemark}{\automaton{Mark}}
\newcommand{\Tmultimark}{\automaton{Multi}}

\newcommand{\dayone}{W_\text{day-one}}

\newcommand\restr[2]{{
  \left.\kern-\nulldelimiterspace 
  #1 
  \vphantom{\big|} 
  \right|_{#2} 
  }}

\makeatletter
\def\ps@headings{%
	\def\@evenfoot{\rlap{\thepage}\hfill\hbox to0pt{\hss{}\hss}\hfill}%
      \def\@oddfoot{\hfill\hbox to0pt{\hss\footnotesize\MakeUppercase{\rightmark}\hss}\hfill
	\llap{\thepage}}%
      \let\@evenhead\@empty%
      \let\@oddhead\@empty%
      \let\@mkboth\markboth
    \def\chaptermark##1{%
      \markright {\MakeLowercase{##1}}}%
    \def\sectionmark##1{}}%
\makeatother

\theoremstyle{theorem}
\newtheorem{theorem}{Theorem}
\newtheorem{lemma}[theorem]{Lemma}

\theoremstyle{definition}

\author{Pierre Lairez}
\affiliation{\institution{Inria, Université Paris Saclay}\city{Palaiseau 91120}\country{France}}
\email{pierre.lairez@inria.fr}

\author{Aleksandr Storozhenko}
\affiliation{\institution{\'Ecole polytechnique} \city{Palaiseau 91120} \country{France}}
\email{aleksandr.storozhenko@polytechnique.edu}

\title{Conway's cosmological theorem and automata theory}
\begin{document}

\begin{abstract}
  John Conway proved that every audioactive sequence (a.k.a. \emph{look-and-say}) decays into a compound of 94~elements, a statement he termed the \emph{cosmological theorem}. The underlying audioactive process can be modeled by a finite-state machine, mapping one sequence of integers to another. Leveraging automata theory, we propose a new proof of Conway's theorem based on a few simple machines, using a computer to compose and minimize them.
\end{abstract}

\maketitle

\section{Introduction}


In 1986, John Conway published his study of integer decay under \emph{audioactive} derivation \cite{Conway_1986,Conway_1987}, veiled in a brilliant atomic metaphor.
The derivation process, now well-known in recreational mathematics, mimics how we read strings.
For instance, the seed “55555”, which we read as “five fives”, is derived to “55”. In turn, “55" yields “25", iteratively generating the audioactive sequence:
\[ 55555 \to 55 \to 25 \to 1215 \to 11121115 \to 31123115 \to \dotsc \]
Any string of numbers, or \emph{word}, such as~“55555”, may be the start of an audioactive sequence\footnote{The strings appearing here should be understood as a sequence of numbers, which most of the time are digits. For example, we note that the audioactive derivation of~“2222222222” should really be the two-element sequence~$(10,2)$, not~$(1,0,2)$.}.

We say that a word \emph{splits} as the concatenation~$uv$ of two contiguous subwords~$u$ and~$v$ if, for all~$k \geq 0$, the $k$th audioactive derivation of~$uv$ is equal to the concatenation of the $k$th derivations of the subwords.
A nonempty word that does not split is called an \emph{atom}, and it is clear that every nonempty word either splits into atoms or is an atom itself. There exist infinitely many distinct atoms.
While most only emerge through the derivation of carefully selected seeds, Conway identified exactly 92 atoms, termed the \emph{common elements}, that appear in the derivation sequence of every word except ``22'' and the empty word.
He further identified two families of atoms, the \emph{transuranic elements}, which appear in the derivation of all words containing a digit~$d \geq 4$.

Celebrated by Conway as the finest achievement of “audioactive chemistry”, the \emph{cosmological theorem} states that there exists some~$N\geq 0$, such that for every word~$x$ and all~$k\geq N$, the $k$th audioactive derivation of~$x$
splits into common and transuranic elements.
An interesting corollary is the \emph{arithmetical theorem}: the length of the $k$th derivation of any given nonempty word, other than ``22'', exhibits geometric growth with ratio $\lambda \simeq 1.303557$, an explicit algebraic number (see \cite{Conway_1986} for more details).
The original proofs of the cosmological theorem, by Conway, Richard Parker, and Mike Guy, claiming a bound $N = 24$, has been lost, but complete proofs have since been given \cite{EkhadZeilberger_1997,Litherland_2003,Watkins_2006}.
\looseness=-1

The audioactive derivation is (almost) described by a type of finite-state machine, called a \emph{transducer}. It is beyond question that Conway, and all others who have subsequently studied audioactive decay, knew about automata theory and the associated formulation of the derivation process. Yet, none of the published proofs make use of it. We propose to fill this gap, leading to a very simple proof of the cosmological theorem based on two cornerstone results of automata theory: first, the composition of two transducers is a transducer; second, there exists an algorithm to check whether two automata recognize the same language (see Section~\ref{sec:basics-autom-theory}).

Our proof strategy closely follows that of Conway, but with a substantial modernization: whereas Conway's method relied on manually ``tracking a few hundred cases'' \cite[p.~14]{Conway_1986}, we harness the expressive power of automata and transducers, delegating the intricate casework to standard computational tools.
In summary, we begin by constructing an automaton to recognize splittings (Theorem~\ref{thm:splittings}).
From this, we derive a transducer capable of extracting the atoms of a given word.
Using this transducer, we generate an automata that recognizes all possible atoms after a specified number of derivations.
Notably, we find that the automata after $24$ and $25$ derivations are identical (Theorem~\ref{thm:cosmological}), indicating a convergence in the atom structure, thereby proving the cosmological theorem.
This approach involves working with automata and transducers that can have several thousand states.

\section{Basics of automata theory}
\label{sec:basics-autom-theory}
In this section, we present the key elements of automata theory that will be employed to demonstrate the cosmological theorem. For an in-depth review of automata theory, we refer to~\cite{HopcroftMotwaniUllman_2007, Shallit_2008, Sakarovitch_2009, Pin_2021}, for example.

\subsection{Transducers}

Let us define an \emph{alphabet} to be a finite set, calling its elements \emph{symbols}. A \emph{word} over an alphabet~$A$ is a finite sequence of symbols of~$A$. A \emph{subword} is a contiguous subsequence of a word. The set of words over~$A$ is denoted~$A^*$, and a \emph{language} over an alphabet~$A$ is a subset of~$A^{*}$. The empty word is denoted~$\varepsilon$, so we make sure that~$\varepsilon$ never denotes a symbol of the alphabet. We further define~$A^? = A \cup \left\{ \varepsilon \right\}$, the alphabet augmented with $\varepsilon$, denoting the absence of a symbol.

A \emph{transducer} is a finite directed graph whose edges may be labelled by an input and/or an output symbol, and whose states may be labelled with an \emph{initial} and/or \emph{final} tag.
More formally, a transducer is a tuple~$\mathcal{T} = (Q, A_\text{in}, A_\text{out}, E, Q_\text{initial}, Q_\text{final})$ where:
\begin{itemize}
  \item $Q$ is the finite set of \emph{states};
  \item $A_\text{in}$ and~$A_\text{out}$ are the \emph{input} and \emph{output} alphabets, respectively;
  \item $E \subseteq Q \times A_\text{in}^? \times A_\text{out}^? \times Q$ is the set of transitions, made of a source state, an input symbol (or~$\varepsilon$), an output symbol (or~$\varepsilon$), and a target state;
  \item $Q_\text{initial} \subseteq Q$ and~$Q_\text{final}\subseteq Q$ are the sets of \emph{initial} and \emph{final states}, respectively.
\end{itemize}

A transducer~$\mathcal{T}$ defines a \emph{transduction relation}; that is, a binary relation between~$A^*$ and~$B^*$, denoted~$\to_\mathcal{T}$.
We say that~$u \to_\mathcal{T} v$ if there is a path in the graph of the transducer~$\mathcal{T}$
from an initial state to a final state, such that the concatenation of the input (respectively output) symbols at the edges along the path is equal to~$u$ (respectively~$v$). We say that~$u$ is the \emph{input word} and~$v$ is the \emph{output word}.
If we want to interpret the transducer as a machine reading some input---transitioning from one state to another after each symbol,
and producing output on transitions---it is a \emph{nondeterministic} machine:
for a given input word, there may be zero, one, finitely or infinitely many execution paths and output words.
This nondeterminism will be used extensively.
The \emph{input language} of~$\mathcal{T}$, denoted~$L_\text{in}(\mathcal{T})$, is the set of all~$u\in A_\text{in}^*$ such that~$u\to_\mathcal{T} v$ for some~$v\in A_\text{out}^*$.
The \emph{output language} of~$\mathcal{T}$, denoted~$L_\text{out}(\mathcal{T})$, is the set of all~$v\in A_\text{out}^*$ such that~$u\to_\mathcal{T} v$ for some~$u\in A_\text{in}^*$.
We say that~$\mathcal{T}$ \emph{accepts} (respectively \emph{rejects}) a word~$x\in A_\text{in}^*$ if it belongs (respectively does not belong) to~$L_\text{in}(\mathcal{T})$.

Different transducers~$\mathcal{S}$ and~$\mathcal{T}$ may induce the same transduction relation. In this case, we say that they are \emph{equivalent} and write~$\mathcal{S} \equiv \mathcal{T}$.
Note that the equivalence of transducers is undecidable, meaning it cannot be verified by any finite-time algorithm \cite[\S\textbf{3}.5]{Pin_2021}.

\subsection{Useful examples of transducers}
Figures~\ref{fig:multimark}, \ref{fig:singlemark}, \ref{fig:scissors} and~\ref{fig:counter} show examples of transducers that will be useful in the proof of the cosmological theorem, and that we describe below.
Let~$A$ be an alphabet not containing the symbol~$\splt$, and let~$B = A \cup \left\{ \splt \right\}$.

The “multimark” transducer, denoted~$\Tmultimark$ (Figure~\ref{fig:multimark}), works with the input alphabet~$A$ and output alphabet~$B$. It inserts arbitrarily many  $\splt$ symbols nondeterministically in the input word. The relation induced by this transducer is characterized as follows: ~$u \to_{\small\Tmultimark} v$ if and only if~$u$ can be obtained from~$v$ by deleting the~$\splt$ symbols. For example,~$312 \to_{\small\Tmultimark} 3\splt1\splt\splt2\splt\splt$.

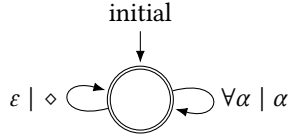
\begin{figure}[th]
  \centering

    \begin{tikzpicture}[shorten >=1pt, node distance=2cm, on grid, auto, initial text=initial, initial where=above, scale=1.0, transform shape, >={Latex[scale=0.945]}]

        \node [state, initial]  (q_0)                {};
        \draw (q_0) circle (11.4pt); 
        \path[->] (q_0) edge [loop right]  node[right] {$\forall\alpha \mid \alpha $} (q_0);
        \path[->] (q_0) edge [loop left]  node[left] {$\varepsilon \mid \splt $} (q_0);

    \end{tikzpicture}

  \caption{%
    The “multimark”, transducer, denoted~$\Tmultimark$.
    The edges are labelled with the convention ``input symbol $\mid$ output symbol''.
    The notation $\forall \alpha$ for the input symbol means the corresponding edge
    should be duplicated for each symbol in the input alphabet.
    When the output symbol is~$\alpha$, it means “copy the input symbol”.
    The “initial” arrow marks initial state(s). The double stroke marks final state(s).
  }
  \label{fig:multimark}
\end{figure}

The “single mark” transducer, denoted~$\Tsinglemark$ (Figure~\ref{fig:singlemark}), works with the input alphabet~$A$ and the output alphabet~$B$. It nondeterministically inserts a $\splt$ symbol somewhere in the input word, not before the first symbol, and not after the last one. It also accepts the empty word. The relation induced by this transducer is characterized as follows: $\varepsilon\to_{\small\Tsinglemark} \varepsilon$ and~$uv \to_{\small\Tsinglemark} u\splt v$ for any two nonempty words~$u$ and~$v$ not containing~$\splt$.

\begin{figure}[th]
  \centering
  \begin{tikzpicture}[shorten >=1pt,node distance=3cm,on grid,auto,initial text=initial, >={Latex[scale=0.9]}]

    \node[state, initial]  (q_0)                {first};
    \draw (q_0) circle (12.1pt);
    \node[state]          (q_1) [right=of q_0] {\phantom{r}left of \splt};
    \node[state]          (q_2) [right=of q_1] {right of \splt};
    \node[state]          (q_3) [right=of q_2] {last};
    \draw (q_3) circle (11.6pt);

    \path[->] (q_0) edge node[pos=0.48, above, sloped] {$\forall\alpha \mid\hspace{-0.1em} \alpha$}(q_1);
    \path[->] (q_1) edge node[pos=0.48, above, sloped] {$\varepsilon \mid\hspace{-0.1em} \splt$} (q_2);
    \path[->] (q_2) edge node[pos=0.48, above, sloped] {$\forall\alpha \mid\hspace{-0.1em} \alpha$} (q_3);
    \path[->] (q_1) edge [loop above, out = 107.5, looseness = 7]  node[above] {$\forall\alpha \mid \alpha $} (q_1);
    \path[->] (q_2) edge [loop above, out = 107.5, looseness = 7]  node[above] {$\forall\alpha \mid \alpha $} (q_2);
  \end{tikzpicture}

  \caption{The “single mark” transducer, denoted~\Tsinglemark.}
  \label{fig:singlemark}
\end{figure}
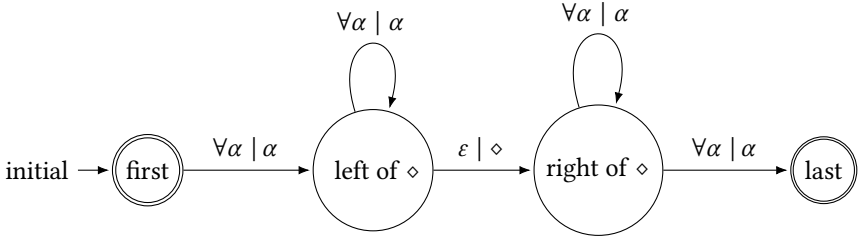

The “scissors” transducer, denoted \Tscissors{} (Figure~\ref{fig:scissors}),
works with the input alphabet~$B$ and the output alphabet~$A$.
It extracts from the input word a substring delimited by~$\splt$ symbols.
The relation induced by this transducer is characterized by~$u \splt v \splt w \to_{\small\Tscissors} v$ for all words~$u, w\in B^*$ and~$v\in A^*$.

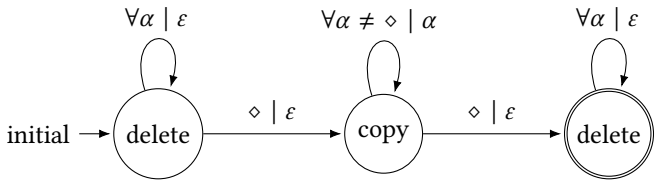
\begin{figure}[ht]
  \centering
  \begin{tikzpicture}[shorten >=1pt,node distance=3cm,on grid,auto,initial text=initial, >={Latex[scale=0.9]}]

    \node[state,initial]  (q_0)                {delete};
    \node[state]          (q_1) [right=of q_0] {copy};
    \node[state]          (q_2) [right=of q_1] {delete};
    \draw (q_2) circle (15.88pt);

    \path[->] (q_0) edge node[pos=0.48, above, sloped] {$\splt \mid \varepsilon$} (q_1);
    \path[->] (q_1) edge node[pos=0.48, above, sloped] {$\splt \mid \varepsilon$} (q_2);
    \path[->] (q_0) edge [loop above, out = 107.5, looseness = 7]  node[above] {$\forall\alpha \mid \varepsilon $} (q_0);
    \path[->] (q_1) edge [loop above, out = 107.5, looseness = 8.5]  node[above] {$\hspace{-0.2em} \forall\alpha \neq \splt \mid \alpha$} (q_1);
    \path[->] (q_2) edge [loop above, out = 107.5, looseness = 7]  node[above] {$\forall\alpha \mid \varepsilon $}  (q_2);
  \end{tikzpicture}

  \caption{The “scissors” transducer, denoted~\Tscissors.}
  \label{fig:scissors}
\end{figure}

Lastly, given~$a\in A$, the “bounded $a$-counter” transducer, denoted~\Tcounter{a}, works with the input alphabet~$A$ and the output alphabet~$A \cup \left\{ 1,2,3 \right\}$. It counts occurrences of~$a$, up to~3.
More precisely,
the relation induced by this transducer is finite and contains only the following ordered pairs:
$a \to_{\small\Tcounter{a}} 1a$,
$aa \to_{\small\Tcounter{a}} 2a$, and
$aaa \to_{\small\Tcounter{a}} 3a$.

\begin{figure}[ht]
  \centering

    \begin{tikzpicture}[shorten >=1pt,node distance=3cm,on grid,auto,initial text=initial, >={Latex[scale=0.9]}]

        \node[state,initial] (i) {};
        \node[state] (a11) [above=2cm of i] {};
        \node[state] (a21) [right=of i] {};
        \node[state] (a31) [below=2cm of i] {};
        \node[state] (a32) [below=2cm of a21] {};
        \node[state] (f) [right=of a21] {};
        \draw (f) circle (11.25pt);

        \path[->] (i) edge node[left] {$a \mid 1$} (a11)
                       edge node[above] {$a \mid 2$} (a21)
                       edge node[left] {$a \mid 3$} (a31);

        \path[->] (a11) edge[bend left=20] node[above] {$\varepsilon \mid a$} (f); 
        \path[->] (a21) edge node[above] {$a \mid a$} (f);
        \path[->] (a31) edge node[above] {$a \mid a$} (a32);
        \path[->] (a32) edge[bend right=25] node[below right] {$a \mid \varepsilon$} (f); 

    \end{tikzpicture}

  \caption{The “bounded $a$-counter” transducer, denoted~\Tcounter{a}.}
  \label{fig:counter}
\end{figure}
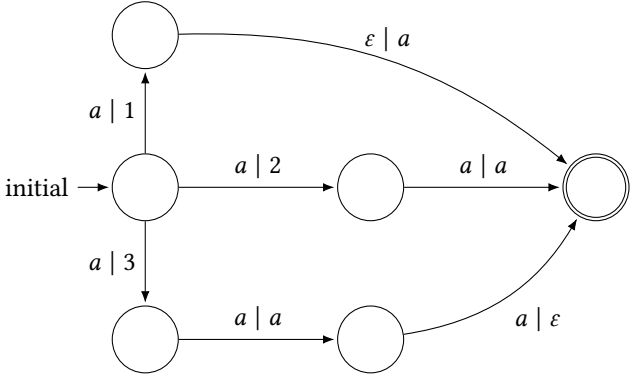

\subsection{Composition of transducers}
Let~$A$, $B$, and~$C$ be alphabets. Let~$\mathcal{U}$ be a transducer from the alphabet~$A$ to~$B$, and let~$\mathcal{V}$ be a transducer from~$B$ to~$C$.
We define, up to equivalence~$\equiv$, a composed transducer~$\mathcal{W}$, such that for all~$x\in A^*$ and~$y\in C^*$, \[ x \to_\mathcal{W} y \ \Leftrightarrow\ \exists \, z\in B^*, x \to_{\mathcal{U}} z \text{ and } z\to_{\mathcal{V}} y, \] see~\cite[Theorem~\textbf{3}.2.2]{Pin_2021}. We denote $\mathcal{W}$  as~$\mathcal{V} \circ \mathcal{U}$. If~$\mathcal{U}$ and~$\mathcal{V}$ induce partial functions (meaning that there is at most one output word for a given input word), the composition~$\mathcal{V}\circ\mathcal{U}$ also induces a partial function which is the composition of the previous ones. The powering notation~$\mathcal{U}^n$ denotes the $n$-fold composition~$\mathcal{U} \circ \dotsb \circ \mathcal{U}$. A detailed description of the construction is, subsequently, given in the implementation section~[\ref{section: Implementation}].

\subsection{Generators, recognizers, and filters}

Let~$\mathcal{T}$ be a transducer with an input alphabet~$A$ and an output alphabet~$B$. If the input alphabet $A$ is empty, we call~$\mathcal{T}$ a \emph{generator}. It is easy to see that the input language of~$\mathcal{T}$ corresponds to the singleton~$\{\varepsilon\}$. Consequently, on the unique input~$\varepsilon$, the transducer generates the complete output language. An example of a \emph{generator} is the ``source'' transducer~\Tgen{}, producing~$B^*$ (Figure~\ref{fig:sink-and-source}). Similarly, an empty output alphabet $B$ implies the only possible output word is~$\varepsilon$. In this case, we call~$\mathcal{T}$ a \emph{recognizer}: on every input~$u \in A^*$, $\mathcal{T}$ either rejects~$u$, or accepts it with output word~$\varepsilon$. An example is given by the ``sink'' transducer~\Trecog, recognizing~$A^*$. Finally, if $A=B$, and if on each transition of~$\mathcal{T}$ the input matches the output symbols, we say that~$\mathcal{T}$ is a \emph{filter}: on input~$u$, $\mathcal{T}$ either rejects~$u$, or accepts it with output word~$u$.

\begin{figure}[H]
  \centering
  \begin{tikzpicture}[shorten >=1pt,node distance=3cm,on grid,auto,initial text=initial, >={Latex[scale=0.9]}]

    \node[state,initial]  (sink)  {};
    \draw (sink) circle (11.5pt);
    \path[->] (sink) edge [loop above, out = 107.5, looseness = 7]  node[above] {$\forall\alpha \mid \varepsilon $} (sink);

    \node[state,initial]  (source)  at (4,0) {};
    \draw (source) circle (11.5pt);
    \path[->] (source) edge [loop above, out = 107.5, looseness = 7]  node[above] {$\varepsilon \mid \forall\alpha$} (source);

  \end{tikzpicture}

  \caption{The ``sink'' \Trecog{} and ``source'' \Tgen{} transducers.}
  \label{fig:sink-and-source}
\end{figure}
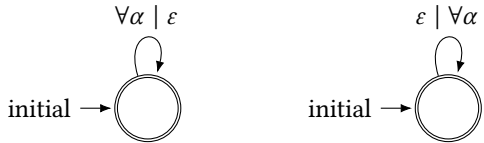

We can convert a recognizer or generator into a filter, and vice versa.
For example, if~$\mathcal{T}$ is a recognizer, its transitions have the form~$(q_1, a, \varepsilon, q_2)$,
which replaced by~$(q_1, a, a, q_2)$ turn~$\mathcal{T}$ into a filter with the same input language.
Furthermore, the composition~$\Trecog \circ \mathcal{T}$ ``deletes'' the output of~$\mathcal{T}$,
turning it into a recognizer for~$L_\text{in}(\mathcal{T})$;
similarly, the composition~$\mathcal{T} \circ \Tgen$ has the effect of feeding~$\mathcal{T}$  all possible inputs,
yielding a generator for~$L_\text{out}(\mathcal{T})$.

The transduction relation of generators (respectively recognizers and filters) is entirely characterized by the output (respectively input) language.
These kinds of transducers are essentially the same concept: the concept of \emph{automata}.
(In the terminology of \cite{HopcroftMotwaniUllman_2007}, our automata are “nondeterministic finite automata with $\varepsilon$-transitions”, or $\varepsilon$-NFA).
Contrary to general transducers, the equivalence of automata is decidable (see below). This is a key property of our proof.

\subsection{Minimal deterministic recognizers}

A recognizer~$\mathcal{T}$ is \emph{deterministic}
if:
\begin{itemize}
  \item it has a single initial state;
  \item it has no~$\varepsilon$-input transitions;
  \item for a given state~$s$ and a given symbol~$a$, there is at most one transition from~$s$ with the input symbol~$a$.
\end{itemize}
This corresponds to the usual definition of the \emph{deterministic finite automaton} (DFA).
Every recognizer is equivalent (in the sense of~$\equiv$)
to a deterministic recognizer through the power set construction \cite[\S2.3.5]{HopcroftMotwaniUllman_2007}.
Moreover, among all deterministic recognizers of a given language \( L \subseteq A^* \), there exists a unique one with the minimal number of states, up to state relabeling \cite[\S 4.4.4]{HopcroftMotwaniUllman_2007}. Given a recognizer for~$L$, there exist various algorithms to compute the associated minimal recognizer. We implemented Brzozowski's algorithm~\cite{Brzozowski_1962} \cite[Chapter~10]{Pin_2021} because of its simplicity. Due to the uniqueness of minimal automata, we can decide the equivalence of two recognizers, \(\mathcal{T}\) and \(\mathcal{T}'\), by checking the equivalence of their corresponding minimal recognizers.

\section{The cosmological theorem}
The first step in proving the cosmological theorem is the formulation of the audioactive derivation, in terms of a transducer. Through the subsequent study of \emph{splittings}, the proof becomes a low-hanging fruit.

\subsection{The audioactive transducer}
Let~$\mathbb{N}_>$ denote the set of positive integers and let $\mathbb{N}_>^*$ be the set of all finite sequences over $\mathbb{N}_>$.
Audioactive derivation is then a map~$C : \mathbb{N}_>^*\to \mathbb{N}_>^*$. A sequence obtained after $n$ applications of $C$, denoted as $x \in C^n(\mathbb{N}_>^*)$, is called a day-$n$ sequence.

The map $C$ cannot be induced by a transducer, as the associated input and output alphabets, $\mathbb{N}_>$, are not finite. Besides this trivial reason, a transducer has a finite number of states, and thus cannot count to arbitrarily large values. We remark, however, that we can make use of the ``one-day theorem''~\cite[p.~10]{Conway_1986}.

\begin{theorem}[One-Day Theorem]\label{thm:day-one-theorem}
    No day-$one$ sequence~$x \in C(\mathbb{N}_>^*)$~contains four consecutive equal symbols, that is no~``aaaa'' subword.
\end{theorem}

\begin{proof}
    By definition of the map $C$, all pairwise consecutive odd positions in $x$, indexed from $0$, must differ. A subword of length four, however, would contain two consecutive equal odd positions.
\end{proof}

After the first audioactive derivation,  all subsequent derivations will only need to count up to three.
We still have the problem of an infinite alphabet, but it is only apparent. Indeed, we have seen that a day-one word~$x$ contains no~$aaaa$ subword,
so its derivation contains no~$44$,~$55$, or~$aa$ subwords with~$a \geq 4$,
as one of the two symbols comes from counting consecutive occurrences of the same symbol in~$x$. Therefore, symbols~$d \geq 4$ never mutually interact past the first derivation and are simply carried over. As such, we simply denote them~$d$.

We, hence, consider the alphabet~$A = \left\{ 1, 2, 3, d \right\}$. Let~$\dayone \subset A^*$ be the set of words not containing any subword of the form~$aaaa$. Audioactive derivation, then, induces a map~$C : \dayone \to \dayone$, entirely described by a transducer (Figure~\ref{fig:audioactive}) that uses~$A$ as both the input and output alphabet.
By virtue of the one-day theorem, it is enough to study the iterations of~$C$ on~$\dayone$ to establish the cosmological theorem.

\begin{figure}[H]
  \centering
  \begin{tikzpicture}[shorten >=1pt,node distance=2cm,on grid,auto,initial text=initial, >={Latex[scale=0.9]}]

    \node[state]  (1)                {};
    \draw (1) circle (11.7pt);    
    \node[state, right=5cm of 1]  (n1)                {};
    \path[->, dotted, thick, >={Latex[scale=0.75]}] (1) edge
    node[below] {\tiny $1\mid 11$\qquad $11 \mid 21$\qquad $111\mid 31$}
    node {\Tcounter{1}} (n1);

    \node[state, below=of 1]  (2)                {};
    \draw (2) circle (11.7pt);    
    \node[state, right=5cm of 2]  (n2)                {};
    \path[->, dotted, thick, >={Latex[scale=0.75]}] (2) edge node {\Tcounter{2}}
    node[below] {\tiny $2 \mid 12$\qquad $22 \mid 22$\qquad $222 \mid 32$}
    (n2);

    \node[state, below=of 2]  (3)                {};
    \draw (3) circle (11.7pt); 
    \node[state, right=5cm of 3]  (n3)                {};
    \path[->, dotted, thick, >={Latex[scale=0.75]}] (3) edge node {\Tcounter{3}}
    node[below] {\tiny $3\mid 13$\qquad $33\mid 23$\qquad $333\mid 33$}
    (n3);

    \node[state, below=of 3]  (d)                {};
    \draw (d) circle (11.7pt);    
    \node[state, right=5cm of d]  (nd)                {};
    \path[->, dotted, thick, >={Latex[scale=0.75]}] (d) edge node {\Tcounter{d}}
    node[below] {\tiny $d \mid 1d$\qquad $dd \mid 2d$\qquad $ddd \mid 3d$}
    (nd);

    \node[state, below right=1.41cm of n1] (n12) {};
    \node[state, initial, below left=1.41cm of 1] (p12) {};
    \node[state, below right=1.41cm of n3] (n3d) {};
    \node[state, initial, below left=1.41cm of 3] (p3d) {};

    \node (mid) at (2.5,-3) {};

    \path (n12.south) edge[in=30, out=-90, looseness=1.5, shorten >= 0pt] (mid.east);
    \path[->] (mid.east) edge[out=210, in=90, looseness=1.5] (p3d.north);

    \path (n3d.north) edge[in=-30, out=90, looseness=1.5, shorten >= 0pt] (mid.east);
    \path[->] (mid.east) edge[out=-210, in=-90, looseness=1.5] (p12.south);

    \path[->] (n1) edge (2);
    \path[->] (n2) edge (1);
    \path[->] (n3) edge (d);
    \path[->] (nd) edge (3);
    \path[->] (p12) edge (1);
    \path[->] (p12) edge (2);
    \path[->] (p3d) edge (3);
    \path[->] (p3d) edge (d);
    \path[->] (n1) edge (n12);
    \path[->] (n2) edge (n12);
    \path[->] (n3) edge (n3d);
    \path[->] (nd) edge (n3d);

  \end{tikzpicture}

  \caption{The “audioactive” transducer, denoted~\Taudio, with $28$~states.
    The dotted edges, labelled by a “bounded counter”, should be substituted accordingly, identifying the source state with the initial state of the transducer, and the target state with its final state. All other edges have the~$\varepsilon$ input and output symbols, omitted for legibility. The choice of layout---with a group for the symbols~1 and~2, and another for~3 and~$d$---is only cosmetic, reducing edge crossings.
  }
  \label{fig:audioactive}
\end{figure}
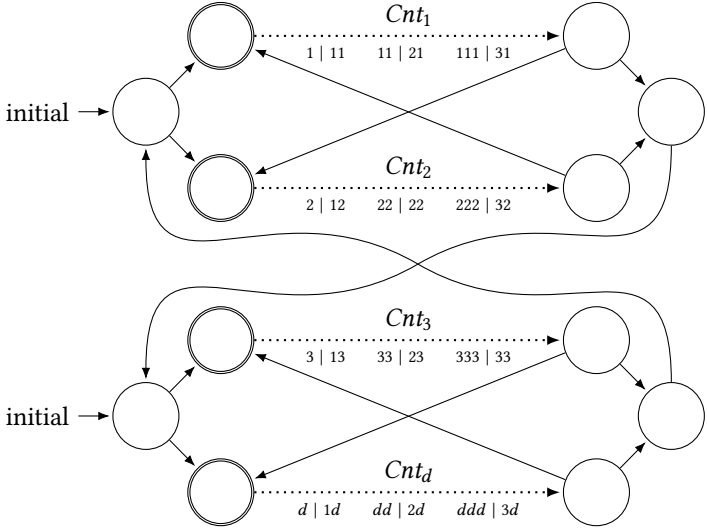

\subsection{Splittings}

Let~$C^n(x)$ denote the~$n$th audioactive derivation of a word~$x \in \dayone$.
A word~$x \in \dayone$ \emph{splits} into~$u_1 \dotsb u_r$ if~$C^n(x) = C^n(u_1) \dotsb C^n(u_r)$ for all~$n \geq 0$, meaning that the~$u_i$ do not interact.
This happens exactly when the last digit of~$C^n(u_1 \dotsb u_i)$ is different from the first digit of~$C^n(u_{i+1} \dotsb u_r)$
for all~$n \geq 0$ and all~$1 \leq i < r$.
We say that each~$u_i$ is a \emph{splitting factor} of~$x$.
A nonempty word which admits a single nontrivial splitting factor is an \emph{atom}.
A splitting factor which is an atom is called an \emph{atomic factor}.
It is easy to check that every nonempty word admits a unique splitting into atoms.
For example:
\begin{itemize}
  \item $32212$ splits into $3 \cdot 2212$,\footnote{This is not trivial! This follows from the observation that all the derivations of~3 end with~3, while all the derivations of~2212 start with 2, see Conway's Starting Theorem \cite{Conway_1986}.
    We can also check that the word “3\splt 2212” is accepted by the recognizer \Rsplit\ introduced below.}
  \item $32213$ splits into $3 \cdot 2213$,
  \item $3221d$ splits into $3 \cdot 221d$,
  \item $32211$ is an atom,\footnote{The only possible splittings would be $3\cdot 2211$, which does not work after two derivations, and~$322\cdot 11$, which does not work after one derivation.}
  \item for all~$n \geq 1$, the~$n$-time concatenation of 332 is an atom,
  showing that there are infinitely many atoms.\footnote{This follows from the cycle formed by the states \texttt{w}, \texttt{n}, and~\texttt{d}
    in the atom recognizer that we will construct below (Table~\ref{tab:Rirred}).}
\end{itemize}
(The first three assertions can be checked with the \Rsplit\ automaton, while the last two can be checked with the \Rirred\ automaton, both introduced below.)

Splitting is subtle. At first glance, there may seem to be infinitely many conditions to check.
Yet, we will see that splittings can be recognized by an automaton.
To this end, we consider the augmented alphabet~$B = \left\{1,2,3,d,\splt\right\}$.
A word~$u_1 \splt \dotsb \splt u_r$ over~$B$ is a \emph{splitting}
if~$u_1 \dotsb u_r \in \dayone$ and $C^n(u_1\dotsb u_r) = C^n(u_1)\dotsb C^n(u_r)$
for all $n\geq 0$. The set of splittings forms a language over $B$, and we now construct its associated recognizer, $\Rsplit$.

The “augmented audioactive” transducer, denoted~\Taa{} (Figure~\ref{fig:audioactive+}),
extends~\Taudio{} by using~$B$ as both the input and output alphabet. Reading~\splt{} in an accept state, the~\Taa{} transducer outputs the same symbol and remains in the same state.
Conversely, when~\splt{} is read in a non-accepting state, the entire input word is rejected. For example, 22\splt22 is rejected, but~$22\splt33 \to_{\Taa} 22\splt 23$.
The input language of~\Taa{} is the set of words with no $aaaa$ subword (with~$a\in A$), and no~$a\splt^+ a$ subword (with~$a\in A$, where~$\splt^+$ means one or more~$\splt$).

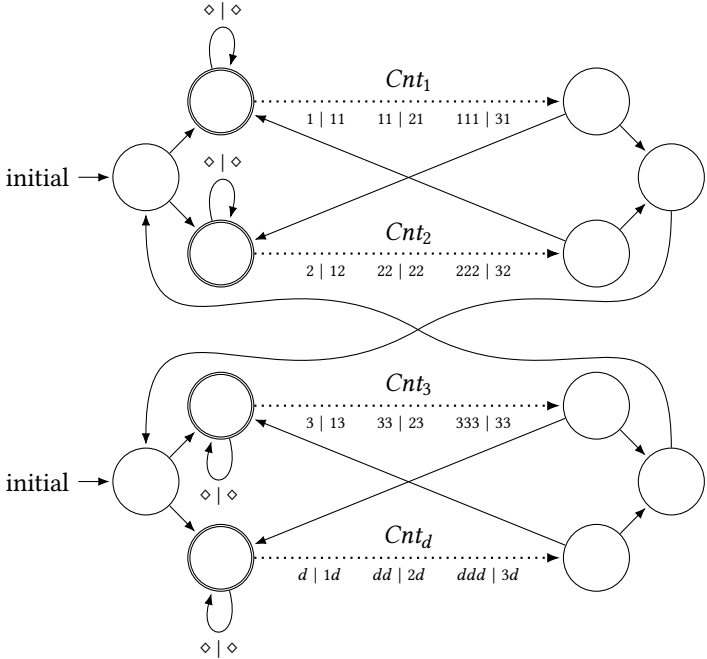
\begin{figure}[th]
  \centering
  \begin{tikzpicture}[shorten >=1pt,node distance=2cm,on grid,auto,initial text=initial, >={Latex[scale=0.9]}]
    \node[state]  (1)                {};
    \draw (1) circle (11.7pt);     
    \node[state, right=5cm of 1]  (n1)                {};
    \path[->, dotted, thick, >={Latex[scale=0.75]}] (1) edge
    node[below] {\tiny $1\mid 11$\qquad $11\mid 21$\qquad $111\mid 31$}
    node {\Tcounter{1}} (n1);

    \node[state, below=of 1]  (2)                {};
    \draw (2) circle (11.7pt);     
    \node[state, right=5cm of 2]  (n2)                {};
    \path[->, dotted, thick, >={Latex[scale=0.75]}] (2) edge node {\Tcounter{2}}
    node[below] {\tiny $2\mid 12$\qquad $22\mid 22$\qquad $222\mid 32$}
    (n2);

    \node[state, below=of 2]  (3)                {};
    \draw (3) circle (11.7pt);  
    \node[state, right=5cm of 3]  (n3)                {};
    \path[->, dotted, thick, >={Latex[scale=0.75]}] (3) edge node {\Tcounter{3}}
    node[below] {\tiny $3\mid 13$\qquad $33\mid 23$\qquad $333\mid 33$}
    (n3);

    \node[state, below=of 3]  (d)                {};
    \draw (d) circle (11.7pt);  
    \node[state, right=5cm of d]  (nd)                {};
    \path[->, dotted, thick, >={Latex[scale=0.75]}] (d) edge node {\Tcounter{d}}
    node[below] {\tiny $d\mid 1d$\qquad $dd\mid 2d$\qquad $ddd\mid 3d$}
    (nd);

    \node[state, below right=1.41cm of n1] (n12) {};
    \node[state, initial, below left=1.41cm of 1] (p12) {};
    \node[state, below right=1.41cm of n3] (n3d) {};
    \node[state, initial, below left=1.41cm of 3] (p3d) {};

    \node (mid) at (2.5,-3) {};

    \path (n12.south) edge[in=30, out=-90, looseness=1.5, shorten >= 0pt] (mid.east);
    \path[->] (mid.east) edge[out=210, in=90, looseness=1.5] (p3d.north);

    \path (n3d.north) edge[in=-30, out=90, looseness=1.5, shorten >= 0pt] (mid.east);
    \path[->] (mid.east) edge[out=-210, in=-90, looseness=1.5] (p12.south);

    \path[->] (n1) edge (2);
    \path[->] (n2) edge (1);
    \path[->] (n3) edge (d);
    \path[->] (nd) edge (3);
    \path[->] (p12) edge (1);
    \path[->] (p12) edge (2);
    \path[->] (p3d) edge (3);
    \path[->] (p3d) edge (d);
    \path[->] (n1) edge (n12);
    \path[->] (n2) edge (n12);
    \path[->] (n3) edge (n3d);
    \path[->] (nd) edge (n3d);

    \path[->] (1) edge [loop above]  node[above] {$\hspace{0.15em}\small \splt \hspace{0.1em} \scriptstyle{\mid} \hspace{0.1em} \splt $} (p12);
    \path[->] (2) edge [loop above]  node[above] {$\hspace{0.15em}\small \splt \hspace{0.1em} \scriptstyle{\mid} \hspace{0.1em} \splt $} (p3d);
    \path[->] (3) edge [loop below]  node[below] {$\hspace{-0.15em}\small \splt \hspace{0.1em} \scriptstyle{\mid} \hspace{0.1em} \splt $}(p12);
    \path[->] (d) edge [loop below]  node[below] {$\hspace{-0.15em}\small \splt \hspace{0.1em} \scriptstyle{\mid} \hspace{0.1em} \splt $} (p3d);
  \end{tikzpicture}

  \caption{The augmented audioactive transducer $\Taa$.}
  \label{fig:audioactive+}
\end{figure}

\begin{lemma}\label{lem:splitting}
  A word $x \in B^*$
  is a splitting if and only if~$\Taa^n$ accepts~$x$ for all~$n\geq 0$.
\end{lemma}

\begin{proof}
  The condition that~$x\in L_\text{in}(\Taa^n)$ for all~$n\geq 0$ means that
  the~$\splt$ symbols (or groups of consecutive~\splt{} symbols) will never lie between two copies of a symbol of~$A$ after any number of audioactive derivations. This is exactly the definition of a splitting.
\end{proof}

\begin{theorem}[Splitting Theorem]\label{thm:splittings}
  The set of splittings is the input language of~$\Taa^9$.
\end{theorem}

\begin{proof}
  Let~$L_n$ denote the input language of~$\Taa^n$, which is the language recognized by~$\Trecog \circ \Taa^{n}$. By Lemma~\ref{lem:splitting}, the set of splittings is the intersection of all~$L_n$ with~$n \geq 1$. Since $\Taa^{n+1} \equiv \Taa \circ \Taa^n$, it follows that~$L_{n+1} \subseteq L_n$ for all~$n \geq 0$.

  We compute $\Trecog \circ \Taa^9$ and~$\Trecog \circ \Taa^{10}$ and verify computationally that they are equivalent. Thus, $\Trecog \circ \Taa^{n} \equiv \Trecog \circ \Taa^9$ for all~$n \geq 9$, leading to the conclusion that~$\cap_{n \geq 1} L_n = L_9$.
\end{proof}

In terms of computation, we compute~$\Trecog \circ \Taa^n$ using the recurrence relation
\[
\Trecog \circ \Taa^{n+1} = \left(\Trecog \circ \Taa^n\right) \circ \Taa,
\]
minimizing the automata at each step. Given that~\Taa\ has 28 states, a naive computation of~$\Trecog \circ \Taa^{10}$ could lead to an automaton with $28^{10}$ states, likely exhausting the memory of a laptop. However, through iterative minimization, the number of states in~$\Trecog \circ \Taa^n$ never exceeds~40 (see Table~\ref{tab:t_mdr_ln}), and the computation time is below~10\,ms on a standard laptop.

\begin{table}[ht]
  \caption{\centering Number of states of  of $\Trecog\circ\Taa^n$ after determinization and minimization.}

  \small\centering
  \begin{tabular}[c]{rrrrrrrrrrr}
    \toprule
    $n$ & {1} & {2} & {3} & {4} & {5} & {6} & {7} & {8} & {$\geq 9$}  \\
    \midrule
    \# states & 13 & 25 & 37 & 40 & 37 & 29 & 28 & 27 & 21\\
    \bottomrule
  \end{tabular}
  \medskip
  \label{tab:t_mdr_ln}
\end{table}

Following Theorem~\ref{thm:splittings}, we define~$\Rsplit \equiv \Trecog \circ \Taa^9$, a recognizer for the language of all splittings. After determinization and minimization, $\Rsplit$ consists of 21 states (see Table~\ref{tab:tsplit}).

\begin{table}[ht]
  \caption{The splitting recognizer \Rsplit{}, after determinization and minimization. The set of states is~$Q = \left\{ \texttt{S}, \texttt{a},\dotsc, \texttt{t} \right\}$, the input alphabet is~$\left\{ 1,2,3, d, \splt\right\}$.}

  \small\centering
  \setlength{\tabcolsep}{2pt}
  \begin{tabular}{l@{\quad}*{21}{>{\ttfamily}c}}\toprule
    state & S & a & b & c & d & e & f & g & h & i & j & k & l & m & n & o & p & q & r & s & t\\
    initial & $\bullet$ &  &   &   &   &   &   &   &   &   &   &   &   &   &   &   &   &   &   &   &  \\
    final  & $\bullet$ & $\bullet$ & $\bullet$ & $\bullet$ & $\bullet$ & $\bullet$ & $\bullet$ & $\bullet$ & $\bullet$ & $\bullet$ & $\bullet$ & $\bullet$ & $\bullet$ & $\bullet$ &   & $\bullet$ &   &   & $\bullet$ & $\bullet$ &  \\
    \midrule
    input $1$  & a & b & c &  & a & a & a & a & a & a & a & a & a &  &  & p & q & c & b & a & \\
    input $2$  & d & d & d & d & e & f &  & d & d & d & d & d & d & n & o &  & f &  & e & d & \\
    input $3$  & g & g & g & g & g & g & g & h & i &  & g & g & g &  &  & r & i &  &  & g & \\
    input $d$  & j & j & j & j & j & j & j & j & j & j & k & l &  &  &  & s & l &  & k & t & l\\
    input $\splt$  & S & m & m & m & o & o & o & m & m & m & l & l & l & m &  & o &  &  & m & l & \\
    \bottomrule
  \end{tabular}

  \label{tab:tsplit}
\end{table}

Conway also provided an explicit description of splittings, which offers an intriguing comparison to our own. Below, we reproduce Conway's splitting theorem \cite[p.~11]{Conway_1986} \emph{verbatim}, intentionally omitting the details of the intricate notations. As Conway himself noted, ``this heap of conventions makes it hard to check the proofs, since they cover many more cases than one naively expects.'' Nevertheless, we offer some insight into how this statement relates to our 21-state automaton, $\Rsplit$.

The exponent~9 in Theorem~\ref{thm:splittings} is optimal, as shown by the word~$3\splt 133$ which has the following sequence of derivations:
\begin{align*}
3 &\splt 1 3 3 \\
1 3 &\splt 1 1 2 3 \\
1 1 1 3 &\splt 2 1 1 2 1 3 \\
3 1 1 3 &\splt 1 2 2 1 1 2 1 1 1 3 \\
1 3 2 1 1 3 &\splt 1 1 2 2 2 1 1 2 3 1 1 3 \\
1 1 1 3 1 2 2 1 1 3 &\splt 2 1 3 2 2 1 1 2 1 3 2 1 1 3 \\
3 1 1 3 1 1 2 2 2 1 1 3 &\splt 1 2 1 1 1 3 2 2 2 1 1 2 1 1 1 3 1 2 2 1 1 3 \\
1 3 2 1 1 3 2 1 3 2 2 1 1 3 &\splt 1 1 1 2 3 1 1 3 3 2 2 1 1 2 3 1 1 3 1 1 2 2 2 1 1 3 \\
1 1 1 3 1 2 2 1 1 3 1 2 1 1 1 3 2 2 2 1 1 3 &\splt 3 1 1 2 1 3 2 1 2 3 2 2 2 1 1 2 1 3 2 1 1 3 2 1 3 2 2 1 1 3.
\end{align*}
So this word is accepted by~$\Taa^8$ but not~$\Taa^9$.

\begin{center}
  \includegraphics[width=\textwidth]{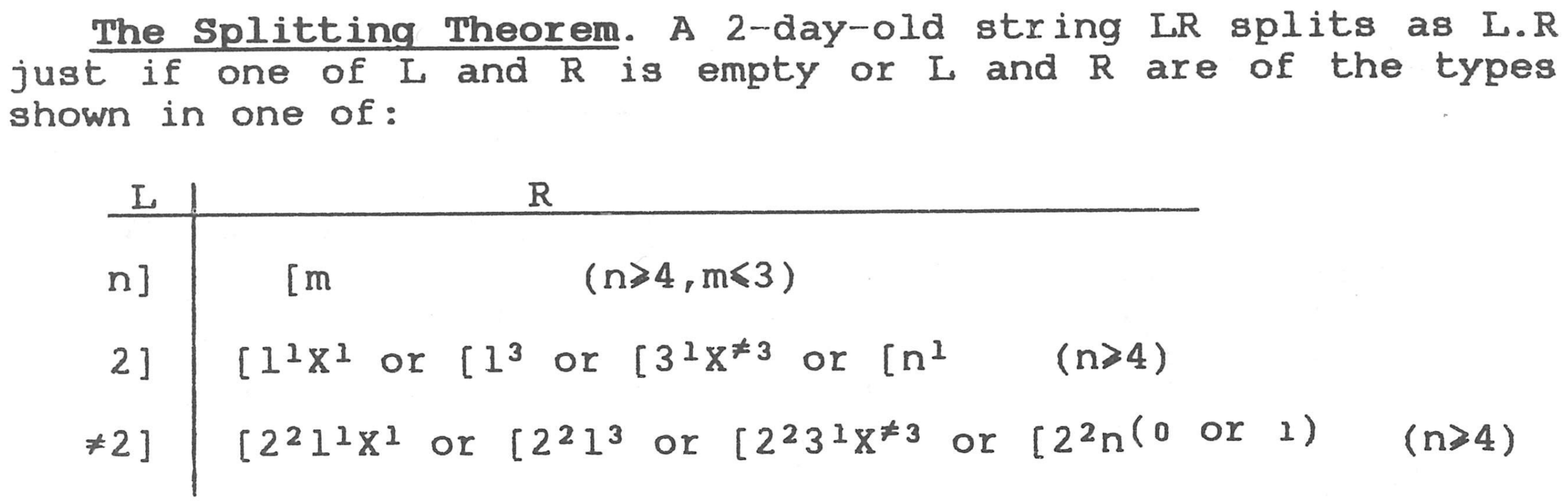}
\end{center}

We briefly outline the structure of the splitting recognizer (Table~\ref{tab:tsplit}). The states~\texttt{S}, \texttt{a}, \ldots, \texttt{l} form a recognizer for the language~$\dayone$. This component of the automaton counts consecutive identical symbols, up to a maximum of three, and rejects all words, containing four or more consecutive equal symbols. For instance, the state~\texttt{c}, reached after reading three consecutive~1s, does not accept an additional~1, as it has no transition for the input~1. When \Rsplit\ encounters the symbol~$\splt$, it transitions to the second part of the automaton, comprised of states~\texttt{l}, \ldots, \texttt{t}. This part has three entry points:
\begin{itemize}
  \item The state~\texttt{l} is reached from the states~\texttt{j}, \texttt{k}, or~\texttt{l}, where the last input received is~$d$. This state corresponds to the first line of Conway's statement and only accepts a digit 1, 2, or~3.
  \item The state~\texttt{o} is reached from the states~\texttt{d}, \texttt{e}, or~\texttt{f}, where the last input received is~2. This state corresponds to the second line of Conway's statement.
  \item The state~\texttt{m} is reached from the states where the last input received is~1 or~3. It corresponds to the third line of Conway's statement. This state only accepts the input~22 and transitions to the state~\texttt{o}, reflecting the similarities in the structure of the second and third lines of Conway's statement.
\end{itemize}

For example, consider the input word~$3\splt 22123$. This matches what Conway denotes as “$\neq 2\rbrack\, \lbrack2^2 1^1 X^1$.” On this input, \Rsplit\ will transition through the states \texttt{S}, \texttt{g}, \texttt{m}, \texttt{n}, \texttt{o}, \texttt{p}, \texttt{f}, or~\texttt{g}. The final state is accepting, indicating that~$3\splt 22123$ constitutes a valid splitting. Next, consider the input word~$3\splt 22122$. This sequence does not match “$\neq 2\rbrack\, \lbrack2^2 1^1 X^1$” because the final~2 cannot be ignored, as it is repeated. (The conventions are truly subtle.) On this input, \Rsplit\ will follow the same sequence of states as before, except at state~\texttt{f}, which will reject the final input symbol~2. Therefore,~$3\splt 22122$ is not a valid splitting. We can explicitly verify the derivation sequence:
\[
3\splt 22122 \to 13\splt 221122 \to 1113\splt 222122 \to 3113\splt 321122,
\]
where we observe the digits before and after the~$\splt$ are identical, thus violating the splitting condition.

From the~\Rsplit{}~recognizer, we can now construct a recognizer for the set of irreducible words.
Consider the recognizer~$\Rsplit \circ \Tsinglemark$
(where~$\Tsinglemark$ is the “single mark” transducer in  Figure~\ref{fig:singlemark}), it accepts exactly the empty word, and words~$uv$, such that~$u\splt v$ is a splitting, with~$u$ and~$v$ nonempty words over~$A$.
Put differently, this recognizer rejects the irreducible words in~$\dayone$, accepting all others.
Let~$\overline{\Rsplit \circ \Tsinglemark}$ denote the complement recognizer, easily computable after determinization \cite[\S4.2]{HopcroftMotwaniUllman_2007}. It accepts exactly the words rejected by~$\Rsplit \circ \Tsinglemark$, recognizing the set of atoms in~$\dayone$ union the set of words not in~$\dayone$.
Finally, we want to also reject the words not in~$\dayone$, which are exactly the words rejected by~$\Trecog \circ \Taudio$.
We, thus, define the irreducible word recognizer
\[ \Rirred{} = \overline{\Rsplit \circ \Tsinglemark} \circ (\Trecog \circ \Taudio)_\textit{filter}, \]
where the subscript “filter” indicates that we turn a recognizer into a filter.
This 26-state automaton is shown in Table~\ref{tab:Rirred}.
By construction, we obtain the following statement.

\begin{lemma}
  The automaton~\Rirred\ recognizes exactly the set of all atoms in~$\dayone$.
\end{lemma}

\begin{table}[ht]
  \caption{\centering The atom recognizer \Rirred, after determinization and minimization.}

  \small\centering
  \setlength{\tabcolsep}{2pt}
  \begin{tabular}{l@{\quad}*{26}{>{\ttfamily}c}}\toprule
    state & S & a & b & c & d & e & f & g & h & i & j & k & l & m & n & o & p & q & r & s & t & u & v & w & x & y\\
    initial & $\bullet$ &   &   &  &   &   &   &   &   &   &   &   &   &   &   &   &   &   &   &   &   &   &   &   &   &  \\
    final  &   & $\bullet$ & $\bullet$ & $\bullet$ & $\bullet$ & $\bullet$ &   &   & $\bullet$ &   &   &   &   &   & $\bullet$ & $\bullet$ & $\bullet$ & $\bullet$ & $\bullet$ &   &   &   &   & $\bullet$ & $\bullet$ & $\bullet$\\
    \midrule
    input $1$  & a & b & c &  & e & c &  & e & e & j & k & c &  &  & a & a &  &  &  &  &  &  &  & a & e & e\\
    input $2$  & x & d & d & d & g & f & g & h &  & l &  &  & m & h & d & d &  &  &  &  &  &  &  & d & y & h\\
    input $3$  & w & w & w & w & i & u &  & i & i & n &  &  &  &  & o &  &  &  &  &  &  & n &  & n & i & i\\
    input $d$  & p & p & p & p & t & v &  & t & t & s &  &  &  &  & p & p & q & r &  & t & r &  & q & p & t & t\\
    \bottomrule
  \end{tabular}
  \label{tab:Rirred}
\end{table}

\subsection{Cosmological Theorem}

The introduced notions enable us to formulate a precise statement about audioactive decay. Let~$E$ denote the set of all common and transuranic elements. It is a set of 94 atoms of which elements are given in the appendix.

\begin{theorem}[Cosmological Theorem]\label{thm:cosmological}
  For every word~$x$ over~$A$ without “$aaaa$” subwords, and all~$n \geq 24$, the $n$th audioactive derivation of~$x$ splits into atoms belonging to~$E$.
\end{theorem}

Let~$E_n$ denote the set of all atomic factors of all the words in $L_\text{out}(\Taudio^n)$.
Our goal is to show that~$E_n$ stabilizes for~$n \geq 24$.
We further aim to explicitly describe the ultimate value of~$E_n$.
The proof rests upon the recognizers~\Rsplit{} and~\Rirred.

Let us construct a transducer~\Tirred{}, for “atomic factor”, over the alphabet~$A$ for both input and output, such that~$u\to_\Tirred v$ if and only if~$v$ is an atomic factor of~$u$.
This transducer is simply
\[ \Tirred = \Rirred_\textit{filter} \circ\Tscissors\circ \Rsplit_\textit{filter}\circ \Tmultimark, \]
where the subscript “filter” indicates the conversion of a recognizer into a filter.
More explicitly, the transducer inputs a word over~$A$, and then:
\begin{itemize}
  \item the multimark transducer~\Tmultimark{}~ nondeterministically inserts~\splt{} symbols;
  \item the splitting recognizer \Rsplit{} exclusively retains the splittings;
  \item the transducer~\Tscissors{} nondeterministically extracts a splitting factor of a given splitting;
  \item the atom recognizer \Rirred{} only retains the atomic factors.
\end{itemize}

\begin{proof}[Proof of Theorem~\ref{thm:cosmological}]
  We compute---see Figure~\ref{fig:TirredTable} and the subsequent discussion---that
  \begin{equation}
    \label{eq:2}
    \Tirred \circ \Taudio^{25} \circ \Tgen \equiv \Tirred \circ \Taudio^{24} \circ  \Tgen,
  \end{equation}
  so~$E_{25} = E_{24}$.
  Since the atomic factors of the derivation of a word~$x$ are exactly the atomic factors of the derivations of the atomic factors of~$x$,
  it follows that the elements of~$E_{n+1}$
  are the atomic factors of the derivation of the elements of~$E_n$.
  Therefore, $E_{25} = E_{24}$ implies~$E_{n} = E_{24}$ for all~$n \geq 24$.

  It remains to enumerate all the elements of~$E_{24}$, which boils down to the enumeration of all the paths from any of the initial states to any of the final states in $\Tirred \circ \Taudio^{24} \circ \Tgen$.
\end{proof}

It is possible to check the equivalence in~\eqref{eq:2} directly, but it is difficult
because $\Taudio^{25} \circ \Tgen$ is a large automaton, with~194,625 states, after minimization.
Instead, we observe that $\Tirred \circ \Taudio \equiv \Tirred \circ \Taudio \circ \Tirred$,
because, as noted in the proof of Theorem~\ref{thm:cosmological},
the irreducible splitting factors of the derivation of a word~$x$ are exactly the irreducible splitting factors of the derivations of the irreducible splitting factors of~$x$.
It follows that for any~$n \geq 1$,
\[ \Tirred \circ \Taudio^{n} \circ \Tgen =  \Tirred \circ \Taudio \circ \left ( \Tirred \circ \Taudio^{n-1} \circ \Tgen \right).\]
This gives a much more efficient recursive way of computing $\Tirred \circ \Taudio^{n} \circ \Tgen$.
Naturally, we minimize the automata after each composition. The maximum number of states for~$\Tirred \circ \Taudio^{n} \circ \Tgen$
is $592$, when~$n=6$. The total computation time is $150$\,ms on a laptop.

\begin{figure}[ht]
  \centering

\begin{tikzpicture}[/tikz/background rectangle/.style={fill={rgb,1:red,1.0;green,1.0;blue,1.0}, fill opacity={1.0}, draw opacity={1.0}}, show background rectangle, scale=0.7]
\begin{axis}[point meta max={nan}, point meta min={nan}, legend cell align={left}, legend columns={1}, title={}, title style={at={{(0.5,1)}}, anchor={south}, font={{\fontsize{14 pt}{18.2 pt}\selectfont}}, color={rgb,1:red,0.0;green,0.0;blue,0.0}, draw opacity={1.0}, rotate={0.0}, align={center}}, legend style={color={rgb,1:red,0.0;green,0.0;blue,0.0}, draw opacity={1.0}, line width={1}, solid, fill={rgb,1:red,1.0;green,1.0;blue,1.0}, fill opacity={1.0}, text opacity={1.0}, font={{\fontsize{8 pt}{10.4 pt}\selectfont}}, text={rgb,1:red,0.0;green,0.0;blue,0.0}, cells={anchor={center}}, at={(1.02, 1)}, anchor={north west}}, axis background/.style={fill={rgb,1:red,1.0;green,1.0;blue,1.0}, opacity={1.0}}, anchor={north west}, xshift={1.0mm}, yshift={-1.0mm}, width={150.4mm}, height={99.6mm}, scaled x ticks={false}, xlabel={$n$}, x tick style={color={rgb,1:red,0.0;green,0.0;blue,0.0}, opacity={1.0}}, x tick label style={color={rgb,1:red,0.0;green,0.0;blue,0.0}, opacity={1.0}, rotate={0}}, xlabel style={at={(ticklabel cs:0.5)}, anchor=near ticklabel, at={{(ticklabel cs:0.5)}}, anchor={near ticklabel}, font={{\fontsize{11 pt}{14.3 pt}\selectfont}}, color={rgb,1:red,0.0;green,0.0;blue,0.0}, draw opacity={1.0}, rotate={0.0}}, xmajorgrids={true}, xmin={0.5}, xmax={25.5}, xticklabels={{$1$,$5$,$10$,$15$,$20$,$25$}}, xtick={{1.0,5.0,10.0,15.0,20.0,25.0}}, xtick align={inside}, xticklabel style={font={{\fontsize{8 pt}{10.4 pt}\selectfont}}, color={rgb,1:red,0.0;green,0.0;blue,0.0}, draw opacity={1.0}, rotate={0.0}}, x grid style={color={rgb,1:red,0.0;green,0.0;blue,0.0}, draw opacity={0.1}, line width={0.5}, solid}, axis x line*={left}, x axis line style={color={rgb,1:red,0.0;green,0.0;blue,0.0}, draw opacity={1.0}, line width={1}, solid}, scaled y ticks={false}, ylabel={number of states}, y tick style={color={rgb,1:red,0.0;green,0.0;blue,0.0}, opacity={1.0}}, y tick label style={color={rgb,1:red,0.0;green,0.0;blue,0.0}, opacity={1.0}, rotate={0}}, ylabel style={at={(ticklabel cs:0.5)}, anchor=near ticklabel, at={{(ticklabel cs:0.5)}}, anchor={near ticklabel}, font={{\fontsize{11 pt}{14.3 pt}\selectfont}}, color={rgb,1:red,0.0;green,0.0;blue,0.0}, draw opacity={1.0}, rotate={0.0}}, ymajorgrids={true}, ymin={0}, ymax={600}, yticklabels={{$0$,$100$,$200$,$300$,$400$,$500$,$600$}}, ytick={{0.0,100.0,200.0,300.0,400.0,500.0,600.0}}, ytick align={inside}, yticklabel style={font={{\fontsize{8 pt}{10.4 pt}\selectfont}}, color={rgb,1:red,0.0;green,0.0;blue,0.0}, draw opacity={1.0}, rotate={0.0}}, y grid style={color={rgb,1:red,0.0;green,0.0;blue,0.0}, draw opacity={0.1}, line width={0.5}, solid}, axis y line*={left}, y axis line style={color={rgb,1:red,0.0;green,0.0;blue,0.0}, draw opacity={1.0}, line width={1}, solid}, colorbar={false}]
    \addplot+[line width={0}, draw opacity={0}, fill={rgb,1:red,0.0;green,0.6056;blue,0.9787}, fill opacity={0.4}, mark={none}, forget plot]
        coordinates {
            (1,43.0)
            (2,138.0)
            (3,266.0)
            (4,409.0)
            (5,534.0)
            (6,592.0)
            (7,570.0)
            (8,513.0)
            (9,430.0)
            (10,361.0)
            (11,320.0)
            (12,310.0)
            (13,308.0)
            (14,278.0)
            (15,248.0)
            (16,255.0)
            (17,258.0)
            (18,266.0)
            (19,277.0)
            (20,273.0)
            (21,267.0)
            (22,258.0)
            (23,251.0)
            (24,243.0)
            (25,243.0)
            (25,0.0)
            (24,0.0)
            (23,0.0)
            (22,0.0)
            (21,0.0)
            (20,0.0)
            (19,0.0)
            (18,0.0)
            (17,0.0)
            (16,0.0)
            (15,0.0)
            (14,0.0)
            (13,0.0)
            (12,0.0)
            (11,0.0)
            (10,0.0)
            (9,0.0)
            (8,0.0)
            (7,0.0)
            (6,0.0)
            (5,0.0)
            (4,0.0)
            (3,0.0)
            (2,0.0)
            (1,0.0)
            (1,43.0)
        }
        ;
    \addplot[color={rgb,1:red,0.0;green,0.6056;blue,0.9787}, name path={41}, area legend, draw opacity={1.0}, line width={1}, solid]
        table[row sep={\\}]
        {
            \\
            1.0  43.0  \\
            2.0  138.0  \\
            3.0  266.0  \\
            4.0  409.0  \\
            5.0  534.0  \\
            6.0  592.0  \\
            7.0  570.0  \\
            8.0  513.0  \\
            9.0  430.0  \\
            10.0  361.0  \\
            11.0  320.0  \\
            12.0  310.0  \\
            13.0  308.0  \\
            14.0  278.0  \\
            15.0  248.0  \\
            16.0  255.0  \\
            17.0  258.0  \\
            18.0  266.0  \\
            19.0  277.0  \\
            20.0  273.0  \\
            21.0  267.0  \\
            22.0  258.0  \\
            23.0  251.0  \\
            24.0  243.0  \\
            25.0  243.0  \\
        }
        ;
    \addplot[color={rgb,1:red,0.0;green,0.6056;blue,0.9787}, name path={42}, only marks, draw opacity={1.0}, line width={0}, solid, mark={*}, mark size={3.0 pt}, mark repeat={1}, mark options={color={rgb,1:red,0.0;green,0.0;blue,0.0}, draw opacity={1.0}, fill={rgb,1:red,0.0;green,0.6056;blue,0.9787}, fill opacity={1.0}, line width={0.75}, rotate={0}, solid}]
        table[row sep={\\}]
        {
            \\
            1.0  43.0  \\
            2.0  138.0  \\
            3.0  266.0  \\
            4.0  409.0  \\
            5.0  534.0  \\
            6.0  592.0  \\
            7.0  570.0  \\
            8.0  513.0  \\
            9.0  430.0  \\
            10.0  361.0  \\
            11.0  320.0  \\
            12.0  310.0  \\
            13.0  308.0  \\
            14.0  278.0  \\
            15.0  248.0  \\
            16.0  255.0  \\
            17.0  258.0  \\
            18.0  266.0  \\
            19.0  277.0  \\
            20.0  273.0  \\
            21.0  267.0  \\
            22.0  258.0  \\
            23.0  251.0  \\
            24.0  243.0  \\
            25.0  243.0  \\
        }
        ;
\end{axis}
\end{tikzpicture}

  \caption{\centering Number of states of~$\Tirred \circ\Taudio^{n} \circ\Tgen$, after determinization and minimization.}
  \label{fig:TirredTable}
\end{figure}
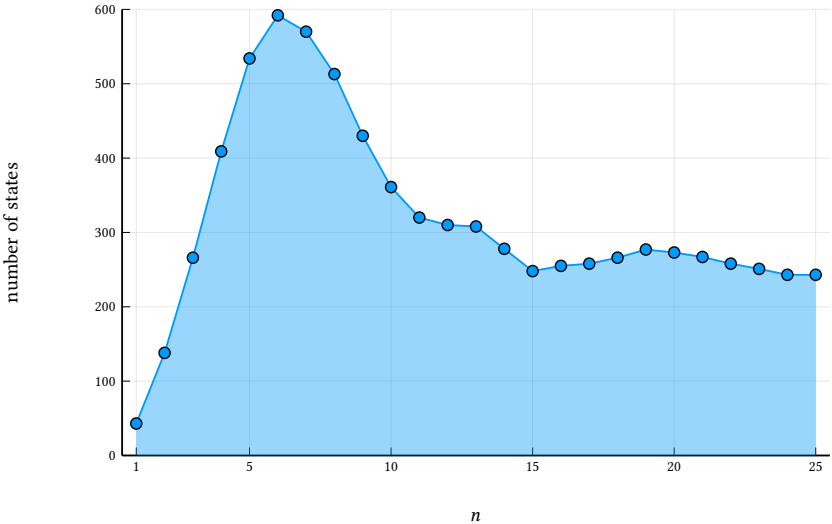



\section{Implementation} \label{section: Implementation} For readers interested in consulting the code or reproducing the results, we present a brief comment on our C++ implementation. The code is available at
\begin{center}
  \url{https://github.com/AleksandrStorozhenko/ConwayTransducer}.
\end{center}

\subsection{Data structure}
It is convenient to index the symbols in the input and output alphabets,
so that we can assume that they are $\left\{ 0,\dotsc,n-1 \right\}$ and~$\left\{ 0,\dotsc,m-1 \right\}$, respectively.
We can represent the transitions of a transducer~$\mathcal{T}$ as a 2d array~$T$ of arrays of pairs
such that~$T[s][a]$ contains the array of all transitions with source state~$s$ and input symbol~$a$. In C++, this gives the following unsophisticated data structure for manipulating transducers.

\begin{verbatim}
#include <set>
#include <vector>
using namespace std;

struct Transducer {
  using symbol = int;
  using state = int;

  int inputSymbols, outputSymbols;
  set<state> startNodes, finalNodes;
  vector<vector<vector<pair<symbol, state>>>> table;
};
\end{verbatim}

The chosen representation does not, however, lend itself to an effective manipulation of generators (with an empty input language), as we cannot efficiently query transitions by output symbol. Instead, we transpose generators into recognizers, swapping input and output symbols for each transition.

\subsection{Composition} Given two finite state transducers $\mathcal{U}$ and $\mathcal{V}$, such that the output alphabet of $\mathcal{U}$ matches the input alphabet of $\mathcal{V}$,
we want to compute a transducer~$\mathcal{W}$ realizing the composion~$\mathcal{V}\circ \mathcal{U}$.

We choose the set of states~$Q_{\mathcal{W}}$ to be~$Q_\mathcal{U} \times Q_\mathcal{V}$. The initial states of~$\mathcal{W}$ are pairs of initial states of~$\mathcal{U}$ and~$\mathcal{V}$ respectively, and similarly for final states.
For each state $(q_u, q_v) \in Q_{\mathcal{U}} \times Q_{\mathcal{V}}$, we declare  the following transitions in~$\mathcal{W}$:
\begin{itemize}
  \item $((q_u, q_v), a, \varepsilon, (q_u', q_v))$ whenever there is a transition~$(q_u, a, \varepsilon, q_u')$ in~$\mathcal{U}$;
  \item $((q_u, q_v), \varepsilon, c, (q_u, q_v'))$ for any transition~$(q_v, \varepsilon, c, q_v')$ in~$\mathcal{V}$; and
  \item $((q_u, q_v), a, c, (q_u', q_v'))$ whenever there exists transitions~$(q_u, a, b, q_u')$ and~$(q_v, b, c, q_v')$ in~$\mathcal{U}$ and~$\mathcal{V}$ respectively for some~$b\in B$.
\end{itemize}

In practice, only a fraction of the constructed states is reachable, so it is worthwhile to construct the state set by a traversal from the initial states to avoid the unreachable sets. 

\subsection{Determinization}
A recognizer~$\mathcal{T}$ admits an equivalent deterministic recognizer~$\mathcal{D}$
with the associated set of states equal to the set of subsets of~$Q_\mathcal{T}$ (the set of states of~$\mathcal{T}$).
The initial state of~$\mathcal{D}$ corresponds to the set of initial states of~$\mathcal{T}$. For~$U \in Q_\mathcal{D}$ and a symbol~$a$, there exists a transition~$(U, a, \varepsilon, V)$ in~$\mathcal{D}$ for~$V \in Q_\mathcal{D}$, the set of all states, reachable by a transition in~$\mathcal{T}$ from a state in~$U$ with an input symbol~$a$.
In practice, only a fraction of these states is reachable. We, thus, perform the determinization by a traversal from the initial state. In the worst case,  the deterministic recognizer~$\mathcal{D}$ admits exponentially many states in the size of~$\mathcal{Q_T}$.

\subsection{Minimization}
Within the context of this paper, minimization of deterministic recognizers (and generators \emph{via} transposition) satisfies two principal aims: computational efficiency via reduction of states; and language equivalence verification.

For minimizing recognizers, we resorted to
Brzozowski's algorithm~\cite{Brzozowski_1962} \cite[Chapter~10]{Pin_2021}:
Given a recognizer~$\mathcal{T}$, swap the source and target states of each transitions (reversal), determinize, reverse again, determinize again.
The algorithm has a worst-case exponential complexity, due to the potential of exponential growth in the number of states at the first determinization,
but it was enough for our purposes.

Algorithms for transducers (minimization, equivalence) are more subtle, with many undecidability results \cite[Chapter~3]{Pin_2021}.
Yet, we found it useful to have a heuristic size-reduction procedure by interpreting a transducer over input alphabet~$A$ and output alphabet~$B$
as a recognizer over the alphabet  $A^? \times B^?$, simply considering the input-output pair of each transition as an input symbol, and minimizing it.
For example, the size of the transducer $\Tirred \circ\Taudio$ reduces from~977 states to~82. This makes the computation of~$\left( \Tirred\circ \Taudio \right)^k\circ \Tgen$ much faster in the proof of Theorem~\ref{thm:cosmological}.

\subsection{Equivalence of recognizers}
To check that two recognizers~$\mathcal{T}$ and~$\mathcal{T}'$ are equivalent, it is enough to check that the minimization of~$\mathcal{T}$ and~$\mathcal{T}'$
\cite[\S 4.4]{HopcroftMotwaniUllman_2007} are equal, up to a bijection~$Q_{\mathcal{T}} \to \mathcal{Q}_{\mathcal{T}'}$.
We can construct this bijection, or prove that it does not exist,
by a traversal from the initial states of~$\mathcal{T}$ and~$\mathcal{T}'$.

\section*{Conclusion}

The study of audioactive decay, from Conway's discovery to the four known
computational proofs, including this one, is an outstanding illustration of the
principles of \emph{experimental mathematics}.
It was, of course, experimentation that led to the formulation of the cosmological theorem, before any proof could be provided.
But the experimental approach in mathematics extends far beyond mere data accumulation for conjecture formation.
Our entire automaton-based proof is experimental in nature.
It is not a proof \emph{by computation}, but a proof \emph{by experimentation}.
While there is an algorithm that proves that~$3\times 19 = 57$, for instance, we do not have an algorithm that proves the cosmological theorem.
Computation serves as a microscope or spectrometer, enabling us to observe and analyze mathematical phenomena without prior knowledge of their nature.
The nature of the cosmological theorem is not a consequence of the formulation of the audioactive derivation by a transducer (choose a different transducer, and you might not observe a splitting theorem or a cosmological theorem), but experimentation and computation uncover this structure.

\begin{acks}
  This work has been supported by the \grantsponsor{anr}{Agence nationale de la
    recher\-che (ANR)}{https://anr.fr}, grant agreement
  \grantnum{anr}{ANR-19-CE40-0018} (De Rerum Natura); and by the
  \grantsponsor{erc}{European Research Council (ERC)}{https://erc.europa.eu}
  under the European Union’s Horizon Europe research and innovation programme,
  grant agreement \grantnum{erc}{101040794} (10000~DIGITS).
\end{acks}

\printbibliography

\section*{Appendix: Periodic Table}

\newcommand{\elemH}{${}_{1}\text{H}$}
\newcommand{\elemHe}{${}_{2}\text{He}$}
\newcommand{\elemLi}{${}_{3}\text{Li}$}
\newcommand{\elemBe}{${}_{4}\text{Be}$}
\newcommand{\elemB}{${}_{5}\text{B}$}
\newcommand{\elemC}{${}_{6}\text{C}$}
\newcommand{\elemN}{${}_{7}\text{N}$}
\newcommand{\elemO}{${}_{8}\text{O}$}
\newcommand{\elemF}{${}_{9}\text{F}$}
\newcommand{\elemNe}{${}_{10}\text{Ne}$}
\newcommand{\elemNa}{${}_{11}\text{Na}$}
\newcommand{\elemMg}{${}_{12}\text{Mg}$}
\newcommand{\elemAl}{${}_{13}\text{Al}$}
\newcommand{\elemSi}{${}_{14}\text{Si}$}
\newcommand{\elemP}{${}_{15}\text{P}$}
\newcommand{\elemS}{${}_{16}\text{S}$}
\newcommand{\elemCl}{${}_{17}\text{Cl}$}
\newcommand{\elemAr}{${}_{18}\text{Ar}$}
\newcommand{\elemK}{${}_{19}\text{K}$}
\newcommand{\elemCa}{${}_{20}\text{Ca}$}
\newcommand{\elemSc}{${}_{21}\text{Sc}$}
\newcommand{\elemTi}{${}_{22}\text{Ti}$}
\newcommand{\elemV}{${}_{23}\text{V}$}
\newcommand{\elemCr}{${}_{24}\text{Cr}$}
\newcommand{\elemMn}{${}_{25}\text{Mn}$}
\newcommand{\elemFe}{${}_{26}\text{Fe}$}
\newcommand{\elemCo}{${}_{27}\text{Co}$}
\newcommand{\elemNi}{${}_{28}\text{Ni}$}
\newcommand{\elemCu}{${}_{29}\text{Cu}$}
\newcommand{\elemZn}{${}_{30}\text{Zn}$}
\newcommand{\elemGa}{${}_{31}\text{Ga}$}
\newcommand{\elemGe}{${}_{32}\text{Ge}$}
\newcommand{\elemAs}{${}_{33}\text{As}$}
\newcommand{\elemSe}{${}_{34}\text{Se}$}
\newcommand{\elemBr}{${}_{35}\text{Br}$}
\newcommand{\elemKr}{${}_{36}\text{Kr}$}
\newcommand{\elemRb}{${}_{37}\text{Rb}$}
\newcommand{\elemSr}{${}_{38}\text{Sr}$}
\newcommand{\elemY}{${}_{39}\text{Y}$}
\newcommand{\elemZr}{${}_{40}\text{Zr}$}
\newcommand{\elemNb}{${}_{41}\text{Nb}$}
\newcommand{\elemMo}{${}_{42}\text{Mo}$}
\newcommand{\elemTc}{${}_{43}\text{Tc}$}
\newcommand{\elemRu}{${}_{44}\text{Ru}$}
\newcommand{\elemRh}{${}_{45}\text{Rh}$}
\newcommand{\elemPd}{${}_{46}\text{Pd}$}
\newcommand{\elemAg}{${}_{47}\text{Ag}$}
\newcommand{\elemCd}{${}_{48}\text{Cd}$}
\newcommand{\elemIn}{${}_{49}\text{In}$}
\newcommand{\elemSn}{${}_{50}\text{Sn}$}
\newcommand{\elemSb}{${}_{51}\text{Sb}$}
\newcommand{\elemTe}{${}_{52}\text{Te}$}
\newcommand{\elemI}{${}_{53}\text{I}$}
\newcommand{\elemXe}{${}_{54}\text{Xe}$}
\newcommand{\elemCs}{${}_{55}\text{Cs}$}
\newcommand{\elemBa}{${}_{56}\text{Ba}$}
\newcommand{\elemLa}{${}_{57}\text{La}$}
\newcommand{\elemCe}{${}_{58}\text{Ce}$}
\newcommand{\elemPr}{${}_{59}\text{Pr}$}
\newcommand{\elemNd}{${}_{60}\text{Nd}$}
\newcommand{\elemPm}{${}_{61}\text{Pm}$}
\newcommand{\elemSm}{${}_{62}\text{Sm}$}
\newcommand{\elemEu}{${}_{63}\text{Eu}$}
\newcommand{\elemGd}{${}_{64}\text{Gd}$}
\newcommand{\elemTb}{${}_{65}\text{Tb}$}
\newcommand{\elemDy}{${}_{66}\text{Dy}$}
\newcommand{\elemHo}{${}_{67}\text{Ho}$}
\newcommand{\elemEr}{${}_{68}\text{Er}$}
\newcommand{\elemTm}{${}_{69}\text{Tm}$}
\newcommand{\elemYb}{${}_{70}\text{Yb}$}
\newcommand{\elemLu}{${}_{71}\text{Lu}$}
\newcommand{\elemHf}{${}_{72}\text{Hf}$}
\newcommand{\elemTa}{${}_{73}\text{Ta}$}
\newcommand{\elemW}{${}_{74}\text{W}$}
\newcommand{\elemRe}{${}_{75}\text{Re}$}
\newcommand{\elemOs}{${}_{76}\text{Os}$}
\newcommand{\elemIr}{${}_{77}\text{Ir}$}
\newcommand{\elemPt}{${}_{78}\text{Pt}$}
\newcommand{\elemAu}{${}_{79}\text{Au}$}
\newcommand{\elemHg}{${}_{80}\text{Hg}$}
\newcommand{\elemTl}{${}_{81}\text{Tl}$}
\newcommand{\elemPb}{${}_{82}\text{Pb}$}
\newcommand{\elemBi}{${}_{83}\text{Bi}$}
\newcommand{\elemPo}{${}_{84}\text{Po}$}
\newcommand{\elemAt}{${}_{85}\text{At}$}
\newcommand{\elemRn}{${}_{86}\text{Rn}$}
\newcommand{\elemFr}{${}_{87}\text{Fr}$}
\newcommand{\elemRa}{${}_{88}\text{Ra}$}
\newcommand{\elemAc}{${}_{89}\text{Ac}$}
\newcommand{\elemTh}{${}_{90}\text{Th}$}
\newcommand{\elemPa}{${}_{91}\text{Pa}$}
\newcommand{\elemU}{${}_{92}\text{U}$}
\newcommand{\elemNp}{${}_{93}\text{Np}$}
\newcommand{\elemPu}{${}_{94}\text{Pu}$}

\begin{longtable}{@{\extracolsep{\fill}}ll>{\footnotesize}l}
  \toprule
  {name} & {derivation} & {\normalsize\rmfamily element}\\ \midrule

  \elemH & \elemH & 22 \\
  \elemHe & \elemHf\  \elemPa\  \elemH\  \elemCa\  \elemLi & 13112221133211322112211213322112 \\
  \elemLi & \elemHe & 312211322212221121123222112 \\
  \elemBe & \elemGe\  \elemCa\  \elemLi & 111312211312113221133211322112211213322112 \\
  \elemB & \elemBe & 1321132122211322212221121123222112 \\
  \elemC & \elemB & 3113112211322112211213322112 \\
  \elemN & \elemC & 111312212221121123222112 \\
  \elemO & \elemN & 132112211213322112 \\
  \elemF & \elemO & 31121123222112 \\
  \elemNe & \elemF & 111213322112 \\
  \elemNa & \elemNe & 123222112 \\
  \elemMg & \elemPm\  \elemNa & 3113322112 \\
  \elemAl & \elemMg & 1113222112 \\
  \elemSi & \elemAl & 1322112 \\
  \elemP & \elemHo\  \elemSi & 311311222112 \\
  \elemS & \elemP & 1113122112 \\
  \elemCl & \elemS & 132112 \\
  \elemAr & \elemCl & 3112 \\
  \elemK & \elemAr & 1112 \\
  \elemCa & \elemK & 12 \\
  \elemSc & \elemHo\  \elemPa\  \elemH\  \elemCa\  \elemCo & 3113112221133112 \\
  \elemTi & \elemSc & 11131221131112 \\
  \elemV & \elemTi & 13211312 \\
  \elemCr & \elemV & 31132\\
  \elemMn & \elemCr\  \elemSi & 111311222112 \\
  \elemFe & \elemMn & 13122112 \\
  \elemCo & \elemFe & 32112 \\
  \elemNi & \elemZn\  \elemCo & 11133112 \\
  \elemCu & \elemNi & 131112 \\
  \elemZn & \elemCu & 312 \\
  \elemGa & \elemEu\  \elemCa\  \elemAc\  \elemH\  \elemCa\  \elemZn & 13221133122211332 \\
  \elemGe & \elemHo\  \elemGa & 31131122211311122113222 \\
  \elemAs & \elemGe\  \elemNa & 11131221131211322113322112 \\
  \elemSe & \elemAs & 13211321222113222112 \\
  \elemBr & \elemSe & 3113112211322112 \\
  \elemKr & \elemBr & 11131221222112 \\
  \elemRb & \elemKr & 1321122112 \\
  \elemSr & \elemRb & 3112112 \\
  \elemY & \elemSr\  \elemU & 1112133 \\
  \elemZr & \elemY\  \elemH\  \elemCa\  \elemTc & 12322211331222113112211 \\
  \elemNb & \elemEr\  \elemZr & 1113122113322113111221131221 \\
  \elemMo & \elemNb & 13211322211312113211 \\
  \elemTc & \elemMo & 311322113212221 \\
  \elemRu & \elemEu\  \elemCa\  \elemTc & 132211331222113112211 \\
  \elemRh & \elemHo\  \elemRu & 311311222113111221131221 \\
  \elemPd & \elemRh & 111312211312113211 \\
  \elemAg & \elemPd & 132113212221 \\
  \elemCd & \elemAg & 3113112211 \\
  \elemIn & \elemCd & 11131221 \\
  \elemSn & \elemIn & 13211 \\
  \elemSb & \elemPm\  \elemSn & 3112221 \\
  \elemTe & \elemEu\  \elemCa\  \elemSb & 1322113312211 \\
  \elemI & \elemHo\  \elemTe & 311311222113111221 \\
  \elemXe & \elemI & 11131221131211 \\
  \elemCs & \elemXe & 13211321 \\
  \elemBa & \elemCs & 311311 \\
  \elemLa & \elemBa & 11131 \\
  \elemCe & \elemLa\  \elemH\  \elemCa\  \elemCo & 1321133112 \\
  \elemPr & \elemCe & 31131112 \\
  \elemNd & \elemPr & 111312 \\
  \elemPm & \elemNd & 132 \\
  \elemSm & \elemPm\  \elemCa\  \elemZn & 311332 \\
  \elemEu & \elemSm & 1113222 \\
  \elemGd & \elemEu\  \elemCa\  \elemCo & 13221133112 \\
  \elemTb & \elemHo\  \elemGd & 3113112221131112 \\
  \elemDy & \elemTb & 111312211312 \\
  \elemHo & \elemDy & 1321132 \\
  \elemEr & \elemHo\  \elemPm & 311311222 \\
  \elemTm & \elemEr\  \elemCa\  \elemCo & 11131221133112 \\
  \elemYb & \elemTm & 1321131112 \\
  \elemLu & \elemYb & 311312\\
  \elemHf & \elemLu & 11132 \\
  \elemTa & \elemHf\  \elemPa\  \elemH\  \elemCa\  \elemW & 13112221133211322112211213322113 \\
  \elemW & \elemTa & 312211322212221121123222113 \\
  \elemRe & \elemGe\  \elemCa\  \elemW & 111312211312113221133211322112211213322113 \\
  \elemOs & \elemRe & 1321132122211322212221121123222113 \\
  \elemIr & \elemOs & 3113112211322112211213322113 \\
  \elemPt & \elemIr & 111312212221121123222113 \\
  \elemAu & \elemPt & 132112211213322113 \\
  \elemHg & \elemAu & 31121123222113 \\
  \elemTl & \elemHg & 111213322113 \\
  \elemPb & \elemTl & 123222113 \\
  \elemBi & \elemPm\  \elemPb & 3113322113 \\
  \elemPo & \elemBi & 1113222113 \\
  \elemAt & \elemPo & 1322113 \\
  \elemRn & \elemHo\  \elemAt & 311311222113 \\
  \elemFr & \elemRn & 1113122113 \\
  \elemRa & \elemFr & 132113 \\
  \elemAc & \elemRa & 3113 \\
  \elemTh & \elemAc & 1113 \\
  \elemPa & \elemTh & 13 \\
  \elemU & \elemPa & 3\\[\medskipamount]

  \multicolumn{3}{l}{\emph{(transuranic elements)}}\\
  \elemNp & \elemHf\ \elemPa\ \elemH\ \elemCa\ \elemPu & 1311222113321132211221121332211$d$ \\
  \elemPu & \elemNp & 31221132221222112112322211$d$ \\
  \bottomrule
\end{longtable}
\end{document}